\documentclass[12pt]{amsart}

\usepackage[margin=2.5cm]{geometry}
\usepackage{tikz}
\tikzset{node distance=2cm, auto}
\usetikzlibrary{matrix,arrows,decorations.pathmorphing}

\usepackage{graphicx}
\usepackage{amsmath,amssymb}

\usepackage{amscd}
\usepackage{verbatim}
\usepackage{enumerate}

\newtheorem{theorem}{Theorem}
\newtheorem{definition}{Definition}
\newtheorem{proposition}{Proposition}

\newtheorem{remark}{Remark}

\newcommand{\g}{\mathfrak{g}}

\newcommand{\hh}{\mathfrak{h}}
\newcommand{\h}{\mathfrak{h}}

\newcommand{\RR}{\mathbb{R}}
\newcommand{\CC}{\mathbb{C}}

\newcommand{\be}{\begin{equation}}
\newcommand{\ee}{\end{equation}}

\newcommand{\cM}{\mathcal{M}}

\newcommand{\cO}{\mathcal O}
\newcommand{\cP}{\mathcal P}

\newcommand{\cB}{\mathcal B}
\newcommand{\cC}{\mathcal C}

\newcommand{\mb}{\mbox}

\newcommand{\stk}{\stackrel}

\title{Degenerately integrable systems}
\author{Nicolai Reshetikhin}
\address{N.R.: Department of Mathematics, University of California, Berkeley,
CA 94720, USA \& KdV Institute for Mathematics, University of Amsterdam,
Science Park 904, 1098 XH Amsterdam, The Netherlands \& ITMO University, Kronverkskii ave. 49, Saint Petersburg 197101, Russia.}
\email{reshetik@math.berkeley.edu}
\begin{document}

\begin{abstract}
The subject of this paper is degenerate integrability in Hamiltonian mechanics.
It starts with a short survey of degenerate integrability.
The first section contains basic notions.
It is followed by a number of examples which include the Kepler system, Casimir models,
spin Calogero models, spin Ruijsenaars models,
and integrable models on symplectic leaves of Poisson Lie groups. The new results are
degenerate integrability of relativistic spin Ruijsenaars and Calogero-Moser systems
and the duality between them.
\end{abstract}
\maketitle

\section*{Introduction}

Degenerately  integrable systems are also known as superintegrable systems and as noncommutative
integrable systems. We will use
the term "degenerate integrability" to avoid possible confusion with supermanifolds,
Lie superalgebras and supergeometry.

Degenerate integrability generalizes well known Liouville integrability of
Hamiltonian systems on a $2n$-dimensional symplectic manifold
to the case when the dimension of invariant tori is $k< n$. When $k=n$
we have the usual Liouville integrability. First examples were known much earlier, see for example \cite{P}\cite{Pauli}\cite{Fock} \cite{Winter}. The notion of degenerate integrability in its modern form was first introduced in \cite{N}. Then a series of examples related to Lie groups was found in \cite{MF}.

First section is a short introduction to degenerate integrability. The rest of the paper is a collection of
examples of degenerately integrable systems. Section two describes the integrability of the Kepler system, which is the classical counterpart of the Bohr model of the hydrogen atom. Its degenerate integrability can be traced back to \cite{P}\cite{Pauli}\cite{Fock}.
The next series of examples, Casimir integrable systems, is described in section 3. These systems can be regarded as degenerations of Gaudin models. They are important for understanding semiclassical asymptotic of q-6j symbols for simple Lie algebras. In sections \ref{CM} and 5 spin Calogero-Moser systems and rational spin Ruijsenaars system are described. The duality between these systems is explained. This section is a concise version of \cite{R2}.
Spin generalization of the Calogero-Moser system was first found in \cite{GH}. The better title of this section would be spin Calogero-Moser-Sutherland-Olshanetsky-Perelomov systems \cite{C}\cite{M}\cite{Su}\cite{OP}. For Liouville integrability of spin Calogero-Moser systems see \cite{KBBT}\cite{LX}. These integrable systems were studied
quite a lot. For the duality in the non-spin case  see \cite{Ru}\cite{Ne}\cite{ER}\cite{FGNR}.
Duality in the context of superintegrability in the non-spin case was futher explored in \cite{AF},
where it was shown that for spinless scattering systems the duality implies the superintegrability.
Duality between relativistic Ruijsenaars and relativistic Calogero-Moser in the non-spin case was studied in
\cite{FK1}\cite{FK2} in the context of Heisenberg double. Further generalization of
Calogero-Moser systems was suggested in \cite{FP}\cite{FP1}. Section \ref{RCM} contains the
proof of degenerate integrability of the relativistic spin Calogero-Moser systems, of the
relativistic spin Ruijsenaars system and the duality between them. Results from this section seem to be new.
The last sections is based on \cite{R2}. It describes the degenerate integrability of Toda type systems on symplectic leaves of simple Poisson Lie groups with standard Poisson Lie structure. The proof of degenerate integrability in
the linearized case was done in \cite{GS}.

The paper was completed while the author was visiting St. Petersburg, ITMO and POMI.
This visit was supported by the project no. 14-11-00598 funded by Russian Science Foundation.
This work was also supported by the NSF grant DMS-0901431 and by the Chern-Simons endowment.
The author is grateful to G. Schrader and to S. Shakirov for helpful discussions.

\section{Degenerate Integrable systems}

\subsection{Degenerate Integrable systems}
An integrable system on a $2n$ dimensional symplectic manifold is called {\it degenerate}
if all the invariant submanifolds have dimension $k < n$. The nondegenerate case when $k=n$
corresponds to the usual Liouville integrability (non-degenerate case).
We will assume $k\leq n$ and for simplicity, some time we will refer to $k=n$ as a particular case of degenerate integrable systems.

\begin{definition} A {\it degenerate integrable system} on a symplectic
manifold $(\cM_{2n}, \omega)$ consists of a Poisson subalgebra $C_J(\cM_{2n})$ in $C(\cM_{2n})$ of rank $2n-k$
which has a Poisson center $C_I(\cM_{2n})$ of rank $k$.
\end{definition}

A Hamiltonian dynamics generated by the function $H\in C(\cM)$ is said to be degenerately integrable if $H\in C_I(M)$.
If $J_1,\dots, J_{2n-k}$ are independent functions from $C_J(\cM)$, we have
\[
\{H,J_i\}=0, \ \ i=1,\dots, 2n-k .
\]
In other words, functions $J_i$ are integrals of motion for $H$.
One can say that Hamiltonian fields generated by $J_i$ describe the symmetry of
the Hamiltonian flow generated by $H$. In this sense, it is natural to call functions from $C_I(\cM_{2n})$ (Poisson commuting) {\it Hamiltonians}, while functions $C_J(\cM_{2n})$ with be called {\it integrals
of motion} for Hamiltonians.

The level surface $\cM(c_1,\dots,c_{2n-k})=\{x\in
\cM | J_i(x)=c_i\}$ of functions $J_i$ is called {\it generic}, relative
to  $C_I(\cM_{2n})$ with  $k$ independent functions $I_1,\dots,I_k\in C_I(\cM_{2n})$ if the form $dI_1\wedge\dots\wedge
dI_{k}$ does not vanish identically on it. Then the following holds (as shown in \cite{N}):

\begin{theorem}\label{deg}
\begin{enumerate}
\item Flow lines of any $H\in C_I(\cM_{2n})$ are parallel to level surfaces of $J_i$.
\item Each connected component of a generic level surface
has canonical affine structure generated by the flow lines of
$I_1,\dots, I_k$.
\item The flow lines of $H$ are linear in this affine
structure.

\end{enumerate}
\end{theorem}

When $k=n$ this theorem reduces to the Liouville integrability.
As a consequence, each generic level surface is isomorphic to $\RR^l\times (S^1)^{k-l}$
for some $0\leq l \leq k$.

The notion of degenerate integrability has a simple
semiclassical meaning. In the Liouville integrable systems
when there are $n$ Poisson commuting integrals on a $2n$ dimensional
symplectic manifold the semiclassical spectrum of quantum integrals is either
non-degenerate or has stable degeneracy which is determined by the number of
connected components of fibers in the Lagrangian fibration given by level surfaces of Hamiltonians.

In degenerate integrable systems the semiclassical
spectrum of quantized commuting integrals $I_i$ is expected to be degenerate
with the multiplicity $h^{n-k}vol(p^{-1}(b))(1+O(h))$. Quantization of the Poisson algebra generated by $J_i$ gives the associative algebra, which describes the symmetry of the joint spectrum of quantum integrals.

Geometrically, a degenerate integrable system consists of
two Poisson projections
\begin{equation}\label{psi-pi}
\cM_{2n}\stk{\pi}{\longrightarrow} P_{2n-k}
\stk{p}\longrightarrow B_{k}
\end{equation}
where $P_{2n-k}$ and $B_k$ are Poisson manifolds and $B_k$ has trivial Poisson
structure. In the algebraic setting, $P_{2n-k}$ is the spectrum (of primitive ideals) of $C_J(\cM)$ and
$B_k$ is the spectrum of $C_I(\cM)$. Fibers of $p$ are (possibly disjoint unions of)
symplectic leaves of $P$.

One should emphasize that degenerate integrability is a special structure which
is stronger than Liouville integrability: invariant tori now have dimension $k<n$.
In the extreme case of $k=1$ all trajectories are periodic. A degenerately
integrable system may also be Liouville integrable, as in the case of spinless Calogero-Moser system
\cite{W}.

The projection $p\circ \pi: \cM\to B_k$ defines the mapping of tangent bundles
$d(p\circ \pi): T\cM\to TB_k$. This gives the distribution
\[
D_B=\omega^{-1}(ker(d(p\circ \pi))^{\perp})\subset T\cM
\]
where the symplectic form $\omega$ is regarded as an isomorphism $T\cM\simeq T^*\cM$
and $ker(d(p\circ \pi))^{\perp}\subset T^*\cM$ is the subbundle orthogonal to $ker(d(p\circ \pi))\subset T\cM$.

\begin{proposition} Leaf of $D_B$ through $x\in \cM$ coincides with $\pi^{-1}(\pi(x))$.
\end{proposition}

We will say that two degenerate integrable systems $(\cM,P, B)$ and $(\cM', P', B')$
are {\it spectrally equivalent} if there is
a collection of mappings
\begin{itemize}
\item $\phi: \cM\to \cM'$, a mapping of Poisson manifolds,
\item $\phi_1: P\to P'$, a mapping of Poisson manifolds,
\item $\phi_2: B\simeq B'$, a diffeomorphism.
\end{itemize}
such that the following diagram is commutative
\begin{equation}\label{deg-int-equivalence}
\begin{tikzpicture}
\node (M) {$\cM$};
  \node (P) [below of=M] {$P$};
  \node (B) [below of=P] {$B$};
  \node (N) [right of=M] {$\cM'$};
  \node (C) [right of=P] {$P'$};
  \node (D) [right of=B] {$B'$};
  \draw[->] (M) to node {$\pi$} (P);
  \draw[->] (M) to node {$\phi$} (N);
  \draw[->] (N) to node {$\pi'$} (C);
  \draw[->] (P) to node {$\phi_1$} (C);
  \draw[->] (P) to node {$p$} (B);
  \draw[->] (C) to node {$p'$} (D);
  \draw[->] (B) to node {$\phi_2$} (D);
\end{tikzpicture}
\end{equation}
Note that the mappings $\phi$ and $\phi_1$ may not be diffeomorphisms.
If they are diffeomorphisms then the systems are called {\it equivalent
degenerately integrable systems}.

\subsection{Degenerately integrable systems and Lagrangian fibrations}
When we have an integrable system on a $2n$-dimensional symplectic manifold $\cM$, which is a Lagrangian fibration
$\tilde{\pi}:\cM\to \cB$,  and a projection  $\pi_2:\cB\to B$ of this fibration to a $k$-dimensional manifold $B$ with $k<n$,
we can construct $P$ as the space of leaves of of the tangent distribution
\[
D=ker(d(\pi_2\circ \tilde{\pi}))^{\perp, \omega}
\]
Here $dp: TM\to TN$ is the differential of $p: M\to N$ and $V^{\perp, \omega}\subset W$ is the
subspace of a symplectic space $W$ which is symplectic orthogonal to $V$.

From $P=\cM/D$ we have a natural projection to $B$ and a natural projection on $\cB$.
If $P$ is smooth, i.e. if the distribution $D$ is integrable, we have a degenerate integrable
system and a commutative diagram (\ref{dint-int}).

\begin{equation}\label{dint-int}
\begin{tikzpicture}
\node (M) {$\cM$};
  \node (P) [below of=M] {$P$};
  \node (B) [below of=P] {$B$};
  \node (C) [right of=P] {$\mathcal{B}$};
  \draw[->] (M) to node {$\pi$} (P);
  \draw[->] (M) to node {$\tilde{\pi}$} (C);
  \draw[->] (P) to node {$\pi_1$} (C);
  \draw[->] (P) to node {$p$} (B);
  \draw[->] (C) to node {$\pi_2$} (B);
\end{tikzpicture}
\end{equation}
An example of such system is the (spinless) Calogero-Moser system \cite{W}.
For other examples see \cite{AF}\cite{FK1}.

\subsection{Action-angle variables}
Degenerate integrable systems admit action-angle variables, see \cite{N}.

For a generic point $c\in P_{2n-k}$ the level
surface $\pi^{-1}(c)$ admits angles coordinates $\varphi_i$. This is an affine coordinate
system generated by the flow lines of Hamiltonian vector fields of integrals $I_1,\dots I_k$ \cite{N}.
In a tubular neighborhood of $p^{-1}(c)$ the symplectic form $\omega$ on $\cM$ can be written as
\[
\omega=\omega_c+\sum_{i=1}^k d\varphi_i\wedge dI_i,
\]
where $\omega_c$ is the symplectic form on the symplectic leave through $c$ in $P_{2n-k}$.

The rest of the paper will focus on specific examples of degenerately integrable
systems.

\section{Kepler system}In this case the phase space is $M=\RR^6$ with coordinates, $p_i, q_i, i=1,2,3$
and with symplectic form
\[
\omega=\sum_{i=1}^3 dp_i\wedge dq^i
\]
The Hamiltonian is
\[
H=\frac{1}{2}p^2-\frac{\gamma}{|q|}
\]
The non-commutative Poisson algebra of integrals is generated by
momenta $M_i$ and components of the Lenz vector $A_i$:
\[
M_1=p_2q^3-p_3q^2, \ \ M_2=p_3q^1-p_1q^3, \ \ M_3=p_1q^2-p_2q^1
\]
\[
A_1=p_2M_3-p_3M_2+\gamma\frac{q^1}{|q|}, \ \ A_2=p_3M_1-p_1M_3+\gamma\frac{q^2}{|q|}, \ \ A_3=p_1M_2-p_2M_1+\gamma\frac{q^3}{|q|}
\]
In vector notations $M=p\times q$ and $A=p\times M+\gamma\frac{q}{|q|}$. Components of $M$ and $A$ have the following Poisson brackets:
\begin{equation}\label{P-brak}
\{M_i, M_j\}=\varepsilon_{ijk}M_k, \ \ \{M_i, A_j\}=\varepsilon_{ijk}A_k, \ \ \{A_i, A_j\}=-2H\varepsilon_{ijk}M_k
\end{equation}
\[
\{H, M_i\}=\{H, A_i\}=0
\]
The momentum vector $M$ and the Lenz vector $A$ satisfy extra relations
\begin{equation}\label{pol-rel}
(M,A)=0, \ \ (A,A)=\gamma^2-2(M,M)H
\end{equation}

Denote by $P_5$ the 5-dimensional Poisson manifold which is a real affine algebraic submanifold in $\RR^7$
with coordinates $M_i, A_i, H$ defined by relations (\ref{pol-rel}) and with Poisson brackets (\ref{P-brak}).

Formulae for $M$, $A$, and $H$ in terms of $p$ and $q$ coordinates describe the Poisson projection $\RR^6\to P_5$.
The following describes level surfaces of $H$ in $P_5$.

The level surface $H=E<0$ is the coadjoint orbit $O_{-E}\subset so(4)^*$. This orbit is isomorphic to $S^2\times S^2$
where each $S^2$ has radius $\gamma/\sqrt{2|E|}$ and $S^2\times S^2$ is naturally embedded into $so(3)^*\times so(3)^*\simeq \RR^3\times \RR^3$. We used the natural isomorphism $so(4)^*\simeq so(3)^*\times so(3)^*$,
where left and right $so(3)^*$ components are given by $L_i=M_i-\frac{A_i}{\sqrt{2|E|}}$ and $R_i=M_i+\frac{A_i}{\sqrt{2|E|}}$.

The level surface $H=0$ is coadjoint orbit in $e(3)^*$ which is isomorphic to $TS^2$ and the sphere has
radius $\gamma$, $(A,A)=\gamma^2$.

The  level surface $H=E>0$ is the hyperboloid $O_E$ which is the coadjoint orbit in $so(3,1)^*$
with natural coordinates $M$ and $B=\frac{A}{\sqrt{2E}}$ and with Casimir functions $(M,B)=0$ and $(B,B)-(M,M)=\gamma^2$.

All of these level surfaces are symplectic manifolds, and we just described symplectic leaves
of the Poisson manifold $P_5$.

This structure correspond to the following sequence of Poisson maps:
\[
\RR^6 \to  P_5 \to \RR
\]
where
\begin{equation}\label{P-Kepl-sl}
P_5\simeq \sqcup_{E<0} S^2\times S^2 \sqcup_{E=0} TS^2 \sqcup_{E>0}  O_E
\end{equation}
The first projection is the map $(p,q)\to (M(p,q), A(p,q), H(p,q))$ and the second
one projects $P_5$ to the $E$-axis.

\section{Casimir integrable systems}
\subsection{Casimir integrable systems}
In this section, $G$ is a complex algebraic group and $\g$ is its Lie algebra.
The phase space of the complex algebraic Casimir system is the Hamiltonian reduction of
the product of coadjoint orbits $\cO_1\times \dots\times \cO_n$
\[
\cM_{\cO_1, \dots, \cO_n}=\{(x_1,\dots, x_n)\in \cO_1\times \dots\times \cO_n|x_1+\dots + x_n=0\}/G
\]
Here we assume that each orbit is regular (passes through a regular element of $\hh^*$).

The coadjoint action of the Lie group $G$ on $\g^*$ is Hamiltonian. The moment map $\cO_{1}\times \dots \times \cO_{n}\to \g^*$ for the diagonal action of $G$ on $\cO_{1}\times \dots \times \cO_{n}$ acts is
\[
(x_1, \dots, x_n)\mapsto x_1+\dots + x_n
\]
It is $G$-invariant, therefore we have a natural map of Poisson manifolds
\begin{equation}\label{tmu}
\widetilde{\mu}: \widetilde{\cM}_{\cO_{1}, \dots,  \cO_{n}}=(\cO_{1}\times \dots \times \cO_{n})/G\to \g^*/{{Ad}^*_G}
\end{equation}
Here we will assume that the quotient space is the GIT quotient. The Hamiltonian reduction gives symplectic leaves of
Poisson manifold $\widetilde{\cM}_{\cO_{1}\times \dots \times \cO_{n}}$:
\[
\cM_{\cO_{1}, \dots , \cO_{n}|\cO_{n+1}}=\mu^{-1}(\cO_{n+1})/G
\]
We have natural symplectomorphisms:
\[
\cM_{\cO_{1}\times \dots \times \cO_{n}|\cO_{n+1}}\simeq \cM_{\cO_{1}, \dots, \cO_{n},-\cO_{n+1}}
\]
and $\cM_{\cO_1, \dots, \cO_n}=\cM_{\cO_1, \dots, \cO_n|\{0\}}$.

Define the Poisson manifold $\cP_{IJ}$ as the fibered product\footnote{Recall that
given two projections $\pi_{1,2}: M_{1,2}\to N$, the fibered product of $M_1$ and $M_2$ over $N$ is
\[
M_1\times_N M_2=\{(x_1,x_2)\in M_1\times M_2| \pi_1(x_1)=\pi_2(x_2)\}
\]
If $\sigma: M_2\to M_2$ is a diffeomorphism, the fibered product twisted by $\sigma$ is
\[
M_1\widetilde{\times}_N M_2=\{(x_1,x_2)\in M_1\times M_2| \pi_1(x_1)=\pi_2(\sigma(x_2))\}
\]}.
\[
\cP_{I, J}=\widetilde{\cM}_{\cO_{i_1},\dots, \cO_{i_k}}\widetilde{\times}_{\g^*/G}\widetilde{\cM}_{\cO_{i_1},\dots, \cO_{i_k}}
\]
where $(I,J)$ is a partition of $(1,\dots, n)$ and the twist is $x\mapsto -x$, and projections
in the fibered product are given by (\ref{tmu}).
The following Poisson maps define the Casimir integrable system in the complex algebraic setting:
\[
\cM_{\cO_1, \dots, \cO_n}\rightarrow \cP_{I,J} \rightarrow \cB_{I,J}\subset  \g^*/Ad^*_G
\]
where $B_{I,J}$ is the image of the last map and the maps are
\begin{multline*}
Ad^*_G(x_1, \dots, x_n)\mapsto (Ad^*_G(x_{i_1},\dots, x_{i_k}), Ad^*_G(x_{j_1},\dots, x_{j_{n-k}}))\mapsto \\
Ad^*_G(x_{i_1}+\dots+ x_{i_k})= Ad^*_G(-x_{j_1}-\dots - x_{j_{n-k}})
\end{multline*}
The variety $B_{I,J}$ has dimension $r$ but it is, generically, smaller then $\g^*/Ad^*_G$.

\subsection{"Relativistic" Casimir systems}
We will keep the same data as in the previous sections. Let $\cC_i\subset G$ be conjugation orbits,
$i=1,\dots, n$. The moduli space of flat $G$-connections on a sphere with $n$ punctures is
a Poisson manifold with the Atiyah-Bott Poisson structure. Fixing conjugacy classes of holonomies around
punctures, gives a symplectic leaf of this Poisson manifold:
\[
\cM_{\cC_1,\dots, \cC_n}=\{(g_1, \dots, g_n)\in \cC_1\times \dots\times \cC_n|g_1\dots g_n=1\}/G
\]
where $G$ acts on the Cartesian product by diagonal conjugations.
The Poisson structure on the moduli space itself, i.e. on $\cM=\{(g_1, \dots, g_n)\in G \times \dots\times G|g_1\dots g_n=1\}/G$ can be described using classical factorizable $r$-matrices as in \cite{FR}.

The group $G$ acts on the product $\cC_1\times \cC_n$ by diagonal conjugations. This action is Poisson and
the mapping
\[
\cC_1\times \cC_n\to G, \ \ (g_1, \dots, g_n)\to g_1\dots g_n
\]
is the group valued moment map for this action \cite{A}. It commutes with the conjugation action of $G$
and gives the Poisson map
\[
\widetilde{\cM}_{\cC_1,\dots, \cC_n}\to G/Ad_G
\]
where
\[
\widetilde{\cM}_{\cC_1,\dots, \cC_n}=\{(g_1, \dots, g_n)\in \cC_1\times \dots\times \cC_n\}/G
\]
As in the previous section, define the Poisson varieties
\[
\cP_{I,J}(\cC_1,\dots, \cC_n)=\widetilde{\cM}_{\cC_{i_1},\dots, \cC_{i_k}}\times_{G/Ad_G} \widetilde{\cM}_{\cC_{j_1},\dots, \cC_{j_{n-k}}}
\]
where $I, J$ is a partition $(1,\dots, n)=I\sqcup J$, and the twisted fibered product
is defined in the previous section. The twist is given by $\sigma: g\mapsto g^{-1}$.

Relativistic Casimir integrable system
is described by the following sequence of Poisson maps
\[
\cM_{\cC_1,\dots, \cC_n}\rightarrow \cP_{I,J}(\cC_1,\dots, \cC_n) \rightarrow \cB_{I,J}(\cC_1,\dots, \cC_n)\subset G/Ad_G
\]
acting as
\[
Ad_G(g_1,\dots, g_n)\mapsto (Ad_G(g_{i_1}\dots g_{i_k}),Ad_G(g_{j_1}\dot sg_{j_{n-k}}))\mapsto [g_{i_1}\dots g_{i_k}]=
[(g_{j_1}\dots g_{j_{n-k}})^{-1}]\in G/Ad_G
\]
Here $\cB_{I,J}$ is the image of the last map, which has dimension $r$ but is, generally, smaller then $\g^*/Ad^*_G$.

\begin{remark} Casimir systems are degenerations of Hitchin systems for a sphere with $n$ punctures.
\end{remark}

\section{Calogero-Moser systems}\label{CM}

\subsection{Degenerate integrability}\label{dint-scm} Spin Calogero-Moser systems are parameterized by pairs $(\g, \cO)$ where $\g$ is a simple Lie algebra
and $\cO$ is a co-adjoint orbit in $\g$. Calogero and Moser discovered such systems for
Lie algebras of type $A$ and coadjont orbit of rank $1$. Sutherland generalized them to trigonometric
and hyperbolic potentials. Olshanetsky and Perelomov generalized them to all simple Lie algebras
and to elliptic potentials. Here we will focus on trigonometric potentials.

The degenerate integrability
of spin Calogero-Moser systems is given by the following collection of Poisson projections.

\[
\begin{tikzpicture}[scale=1.5]
\node (A) at (0,1) {$T^*G$};
\node (B) at (3,1) {$\g^*\times_{\h^*/W} \g^*$};
\node (F) at (6,1) {$\g^*$};
\node (C) at (0,0) {$T^*G/Ad_G$};
\node (D) at (3,0) {$(\g^*\times_{\g^*/G} \g^*)/G$};
\node (E) at (6,0) {$\g^*/G\simeq \g^*/Ad^*_G$};
\path[->,font=\scriptsize,>=angle 90]
(A) edge (B)
([yshift= 1.5pt]B.east) edge node[above] {$L$} ([yshift= 1.5pt]F.west)
([yshift=-1.5pt]B.east) edge node[below] {$R$} ([yshift=-1.5pt]F.west)
(D) edge node[above] {$p$} (E)
(C) edge (D)
(A) edge (C)
(B) edge (D)
(F) edge (E);
\end{tikzpicture}
\]
Here $\g^*\times_{\h^*/W} \g^*$ is the fibered product of two copies of $\g^*$
over $\h^*$. The maps in the upper row of the diagram act as $(x,g)\mapsto (x, -Ad^*_g(x))$, $L(x,y)=x$, and $R(x,y)=y$. Here and below we assume that the co-adjoint bundle $T^*G$ is trivialized
by left translations $T^*G\simeq \g^*\times G$ and has a standard symplectic structure of a cotangent bundle. The lower horizontal sequence of Poisson maps
is at heart of degenerate integrability
of spin Calogero--Moser systems \cite{R2}.

Recall that classical spin Calogero-Moser systems are parameterized by co-adjoint orbits $\cO \subset \g^*$.
For a generic co-adjoint orbit $\cO$,  the phase space of the corresponding spin Calogero-Moser system is
the symplectic leaf $S=\mu^{-1}(\cO)/G$, where $\mu: T^*G\to \g^*$ is the moment map for
the adjoint action of $G$ on $T^*G$:
\[
\mu(x,g)=x-Ad^*_g(x)\in \g^*
\]
Here $x\in \g^*, g\in G$.

The sequence of projections from the diagram above produces the sequence of Poisson
projections
\begin{equation}\label{CM-dint}
S_{\cO}\stackrel{\pi}{\to} P_{\cO} \stackrel{p}{\to} B_{\cO}\subset \g^*/G
\end{equation}
Here
\[
P_{\cO}=N_{\cO}/G, \ \ N_{\cO}=\{(x_1,x_2,x_3)\in \g^*\times \g^* \times \cO|x_1+x_2=x_3\}
\]
Because $N_{\cO}$ is Poisson and the action of $G$ on it is Hamiltonian, the quotient
space $P_{\cO}$ is Poisson. Its dimension is $dim(\g)-2r$ where $r=rank(\g)$ and its symplectic leaves are $\cM(\cO', -\cO', \cO)$. The space $B_{\cO}=\{\cO'\in \g^*/G|\cM_{\cO',-\cO', \cO}\neq \emptyset\}$ has dimension $r=rank(G)$.
Recall that we assume that $\cO$ is generic.

The series of projections (\ref{CM-dint}) describe the degenerate integrability of classical
spin Calogero-Moser model. The Hamiltonian of the classical spin Calogero-Moser system
is the pull-back of the quadratic Casimir from functions on $\g^*/G$ to $S_{\cO}$.
Taking into account that for generic orbit $S_{\cO}\simeq (T^*\h\times \cO//H)/W$, were
$\cO//H$ is the Hamiltonian reduction of $\cO$ with respect to the coadjoint action of
$H$, the Hamiltonian of classical spin Calogero-Moser system can be written as
\[
H_{sCM}=<p,p>+ \sum_{\alpha\in \Delta_+} \frac{\mu_\alpha \mu_{-\alpha}}{(h_{\alpha/2}-h_{-\alpha/2})^2}
\]
where $p, h_\alpha$ are coordinate functions on $T^*\h$ and $\mu_\alpha \mu_{-\alpha}$ is a function on $\cO//H$  see \cite{R2} for details. One can check that the Poisson algebra $C(S_{[t]})$ is isomorphic to
the subalgebra of $W$-invariant functions from $Pol(p, h_\alpha^{\pm 1})\otimes C(\cO_t//H)$ with the Poisson structure
\[
\{p_i, p_j\}=0, \ \ \{p_i, h_\alpha\}=\alpha_i h_\alpha, \ \  \{h_\alpha, h_\beta\}=0
\]
Poisson algebra $C(\cO_t//H)$ of functions on the Hamiltonian  reduction of $\cO_t$ with respect
to the Hamiltonian action of $H$ is the quotient of the Poisson algebra of $H$-invariant functions on
$\cO_t$ with respect to the Poisson ideal generated by Cartan components of $\mu_i$.

Note that the evolution with respect to a central function $F$ on $\g^*$ is quite simple:
\[
(X, g)\mapsto (X, e^{t\nabla F(X)}g)
\]
where $\nabla F$ is the gradient (with respect to the Killing from on $\g$)
of $F$. This formula becomes somewhat complicated after the projection $T^*G\to T^*G/G$.

In the compact case, sequence of projections describing degenerate integrability of Calogero-Moser
system can be written as
\begin{equation}\label{CM-dint}
S_{[t]}\to \sqcup_{[s]\in \h^*/W} \cM_{[s], -[s]|[t]} \to \cB_{[t]}\subset \h^*/W
\end{equation}
Here the moduli space $\cM_{[s_1], [s_2]|[t]}$ is defined as
\[
\cM_{[s_1], [s_2]|[t]}=\{(x_1,x_2)\in \cO_{[s_1]}\times \cO_{[s_2]}| x_1+x_2\in \cO_{[t]}\}/G
\]
and $\cB_{[t]}=\{[s]\in \h^*/W| \cM_{[s], -[s]|[t]}\neq \emptyset \}$. Note that
$\cB_{[t]}$ is unbounded but if $t\neq 0$ it does not contain the vicinity of zero.

\subsection{Rank $1$ orbits for $SL_n$} In this case
\[
\mu_{ij}=\phi_i\psi_j-\delta_{ij}\kappa,
\]
where $\kappa=\frac{1}{n}\sum_{i=1}^n\phi_i\psi_i$.
The Hamiltonian reduction with respect to the action of the Cartan subgroup
introduces constraints $\mu_{ii}=0$ which implies $\phi_i\psi_i=\kappa$.
In this case
\[
\mu_{ij}\mu_{ji}=\phi_i\psi_i\phi_j\psi_j=\kappa^2
\]
which means, in particular, that the Hamiltonian reduction of a rank 1 orbit is a point.
The  spin Calogero-Moser system for such orbits becomes Calogero-Moser system with the Hamiltonian,
which is equal to
\[
H_{CM}=<p,p>+ \sum_{i<j} \frac{\kappa^2}{4\sin(\frac{q_i-q_j}{2})^2}
\]
for the compact real form of $G$.

\section{Rational spin Ruijsenaars systems}

\subsection{Degenerate integrability}
As before, we will assume that  $T^*G$ is trivialized $T^*G\simeq \g^*\times G$ by left translations. Let us denote by $\widetilde{T^*G}$ the Poisson manifold which is $T^*G$ as a manifold, with the Poisson  structure defined uniquely by the following properties:

\begin{itemize}
\item The Poisson algebras $C^\infty({\g}^*)$ with the standard and $C^\infty(G)$ with the trivial Poisson structures respectively, are Poisson subalgebras in $C^\infty(T^*G)$.
\item Poisson bracket between a linear function $X\in{\g}$ on ${\g}^*$
and $f\in C^\infty(G)$ is
\[
\{X,f\} = (L_X-R_X)f
\]
where $L_X$ and $R_X$ are the left and right invariant vector fields
on $G$ generated by $X$.
\end{itemize}

Note that this Poisson structure differs  from the standard
symplectic structure on the cotangent bundle to a manifold.
Symplectic leaves of $\widetilde{T^*G}$ have the form $\cO\times \cC$, where $\cO\subset \g^*$
is a co-adjoint orbit and $\cC\subset G$ is a conjugacy class.

The adjoint action of the group $G$ (the extension of the adjoint action from $G$ to $T^*G$) on $\widetilde{T^*G}$ is Poisson, thus $\widetilde{T^*G}/G$ has a natural Poisson structure. The symplectic leaves of this quotient space are $(\cO\times\cC)/G$
where $G$ acts diagonally on the product.

It is easy to check that the map $T^*G\to \widetilde{T^*G}$ acting as
$\mu\times id : (x,g)\mapsto (x-Ad^*_g(x), g)$, where $\mu$ is the moment map for the adjoint $G$-action,
is Poisson. It is clear that it
commutes with the adjoint $G$-actions.
It induces Poisson map
\begin{equation}\label{map}
T^*G/Ad_G \to \widetilde{T^*G}/Ad_G  \ .
\end{equation}
We also have a natural projection
\[
\widetilde{T^*G}/Ad_G \to G/Ad_G \ .
\]
acting as $Ad_G(x,g)\mapsto Ad_Gg$.
This projection is also Poisson with the trivial Poisson structure on the
base.

Restricting the map (\ref{map}) to the symplectic leaf $S_\cO=\mu^{-1}(\cO)/G$ of
$T^*G/Ad_G$ (see section \ref{dint-scm}), we have
the sequence of Poisson maps describing degenerate integrability of rational spin Ruijsenaars  systems
\[
S_\cO\stackrel{\tilde{\pi}}{\rightarrow} P(\cO) \stackrel{\tilde{p}}{\rightarrow} B(\cO) \subset G/Ad_G \ .
\]
Here $S_\cO=\{(x,g)|x-Ad^*_g(x)\in \cO\}/G$ is the symplectic leaf of $T^*G/Ad_G$ corresponding to the coadjoint
orbit $\cO\in \g^*$, $\tilde{\pi}(G(x,g)=G(x-Ad^*_g(x),g))$, $\tilde{p}(G(x,g))=G(g)$. We have $P(\cO)=\tilde{\pi}(S(\cO))=(\cO\times G)/G\subset \widetilde{T^*G}/G$.
The fiber of the last projection over the conjugation orbit $\cC\in G/Ad_G$ is a symplectic leaf of $P(\cO)$:
\[
P(\cO,\cC)=\{(x-Ad^*_g(x), g)|x\in \cO, g\in \cC\}/G
\]
The space $B(\cO)$ can be described explicitly: $B(\cO)=\{\cC| P(\cO, \cC)\neq \emptyset\}$.
As in the case of the spin Calogero-Moser, the dimension of $B(\cO)$ is $r$, which is the same as the dimension of
a generic fiber of $\tilde{\pi}$.

\subsection{Hamiltonians for $SL_n$ rank 1 orbits}\label{SL_n} Here we assume $G=SL_n$.
In this case we can identify both $\g$ and $\g^*$ with traceless $n\times n$ matrices.
We also assume that $\cO\subset \g^*$ is an orbit through a semisimple element
and that $\mu=x-gxg^{-1}\in \cO$. If we choose the cross-section of the adjoint $G$-action on $T^G$,
where $x_{ij}=\delta_{ij}h_i$,
the symplectic leaf $S(\cO)\in T^*G/G$ (its open dense subset) has coordinates $h_i$, $\mu_{ij}\mu_{ji}$
$g_{ii}$. The Hamiltonian reduction imposes the constraint $\mu_{ii}=0$. Elements $g_{ij}$ satisfy the
equation
\begin{equation}\label{relation}
(h_i-h_j)g_{ij} = \sum_{k=1}^n\mu_{ik}g_{kj} \ .
\end{equation}
We will not try to solve these equations here, in order to find Hamiltonians for $rank>1$ orbits.
In the next section we will do it for rank $1$ case (this computation can also be found in
many other papers, see for example \cite{N}\cite{ER}).

In this case
\[
\mu_{ij}=\phi_i\psi_j-\delta_{ij}\kappa
\]
where $\kappa=<\phi,\psi>/n$ as in the rank $1$ case of Calogero Moser.
The equation (\ref{relation}) implies
\[
(h_i-h_j)g_{ij}=\phi_i\sum_k\psi_kg_{kj}-\kappa g_{ij}
\]
From here we have
\begin{equation}\label{equation-1}
g_{ij}=\frac{1}{h_i-h_j+\kappa}\phi_i\sum_k\psi_kg_{kj}
\end{equation}
This gives the system of equations for $\psi_i\phi_i$
\begin{equation}\label{pp-eq}
\sum_{i=1}^n\frac{\phi_i\psi_i}{h_i-h_j+\kappa}=1
\end{equation}
and the identity
\begin{equation}\label{g-eq}
g_{ii}=\frac{\phi_i}{\kappa}\sum_{k=1}^n\psi_kg_{ki}
\end{equation}
The equation (\ref{pp-eq}) can be solved explicitly:
\[
\phi_i\psi_i=\prod_{j\neq i}\frac{h_i-h_j+\kappa}{h_i-h_j}
\]

Equations (\ref{g-eq}) and (\ref{equation-1}) give the formula for $g_{ij}$
\[
g_{ij} \ = \ \frac{\phi_i\phi_j^{-1}\kappa g_{jj}}{h_i-h_j+\kappa} \ .
\]
Reduced Poisson brackets are log-linear in coordinates $h_i, u_i$\footnote{To be more precise the
algebra of functions on $S(\cO)$ is isomorphic to the algebra of symmetric polynomials in $p_i, u^{\pm 1}$.}
\[
\{h_i, h_j\}=0, \ \ \{h_i, u_j\}=\delta_{ij}, \ \ \{u_i, u_j\}=0
\]
where $u_i$ is related to $g_{ii}$ as
\[
g_{ii}=u_i\prod_{j\neq i}^n\frac{h_i-h_j+\kappa}{h_i-h_j}
\]

The first two elementary $G$-invariant functions of $g$ are
\begin{eqnarray*}
{\mb{tr}}(g) &=& \sum^n_{i=1}g_{ii} \ , \\ {\mb{tr}}(g^2) &=&
\kappa^2\sum_{ij} g_{ii}g_{jj} \
\frac{1}{(h_i-h_j+\kappa)(h_j-h_i+\kappa)} \ .
\end{eqnarray*}
The second function gives the Hamiltonian
of the rational Ruijsenaars system.
\[
H^{rR}=\chi_{\omega_2}(g)=\frac{1}{2}(tr(g^2)-tr(g)^2)=-\sum_{i<j}u_iu_j\prod_{a\in \{ij\}, b\in \{ij\}^\vee}\frac{h_a-h_b+\kappa}{h_a-h_b}
\]
Here $\{i,j\}\subset \{1,\dots,n\}$ and $\{i,j\}^\vee$ is its complimentary subset.
Characters of fundamental representations $\chi_{\omega_i}(g)$ evaluated on elements $g$ described above are classical analogs of rational Macdonald operators.

\subsection{Duality}
A duality relation between (spinless) Calogero-Moser system and (spinless)
rational Ruijsenaars system was
observed in \cite{Ne} \cite{FGNR} (see also references therein).
This is a duality between two Liouville integrable systems which maps angle variables
of one system to the action variable
of the other system. The duality between spin Calogero-Moser and rational spin Ruijsenaars systems
(as the duality of degenerately
integrable systems) was found in \cite{R2}. Here we will recall this property.

Let $F(G(x,\gamma))$ be the fiber of the projection $\pi: T^*G/G\to (\g^*\times_{\g^*/G}\g^*)/G$
containing $G(x,\gamma)$. Recall that  $\pi(G(x,\gamma))=G(x,-Ad^*_\gamma(x))$. It is easy to see that
\[
F(G(x,\gamma))=G(x, Z_x\gamma)
\]
where $Z_x=\{g\in G|Ad_g^*(x)=x\}$. This fiber is the Liouville torus of the spin Calogero-Moser system
passing through the point $G(x,\gamma)$. It projects to
$Ad_G(x)\in \g^*/G$ on the base of the last projection in (\ref{CM-dint}). Hamiltonian flows of
functions on $\g^*/G$ generate angle variable for spin Calogero-Moser system, i.e. an affine coordinate
on $F(G(x,\gamma))$. The generic fiber $F(G(x,\gamma))$ has a dimension $r=rank(G)$.

Define $\widetilde{F}(G(x\gamma))$ as a fiber of the map $\widetilde{\pi}: T^*G/G\to (T^*G,p)/G$
which contains $G(x,\gamma)$. Recall that  $\widetilde{\pi}(x,\gamma)=(x-Ad^*_\gamma(x), \gamma)$. It is easy to see that
\[
\widetilde{F}(G(x,\gamma))=G(x+C_\gamma, \gamma)
\]
Here $C_\gamma=\{x\in \g^*|Ad^*_\gamma(x)=x\}$. This fiber is the Liouville torus of
the rational spin Ruijsenaars system passing through $G(x,\gamma)$. Hamiltonian flows
of functions on $G/Ad_G$ generate an affine coordinate system on it which
is the collection of angle variables for the rational spin Ruijsenaars system.

\begin{theorem}
The fibers $F(G(x,\gamma))$ and $\widetilde{F}(G(x,\gamma))$ are dual in a sense that
\[
F(G(x,\gamma))\cap \widetilde{F}(G(x,\gamma))=G(x,\gamma)
\]
\end{theorem}

For rank $1$ orbits, when both systems are Liouville integrable,
this duality reduces to the one from \cite{Ne}\cite{ER}\cite{FGNR}.

\section{Relativistic spin Calogero-Moser and spin Ruijsenaars systems}\label{RCM}

\subsection{Hamiltonian structure and degenerate integrability of relativistic spin Calogero-Moser and
Ruijsenaars models} The underlying Poisson manifold for relativistic spin Calogero-Moser system is
a "nonlinear" version of $T^*G$ which is known as a Heisenberg double $H(G)$  of $G$ with the standard
Poisson Lie structure. Equivalently, $H(G)/G$ where $G$ acts by diagonal conjugations  can be regarded as the moduli space of flat connections on a punctured torus (see \cite{FR}).

As a manifold, the Heisenberg double is $H(G)=G\times G$. A point $(x,y)$ should be regarded as a pair of
monodromies of the local system on a punctures torus around two fundamental cycles of the torus.
The monodromy around the puncture is $xyx^{-1}y^{-1}$. The Poisson structure on $H(G)$ can be described in terms of
$r$-matrices for standard Poisson Lie structure on $G$. Poisson brackets between coordinate functions
can be written as  (see also \cite{FGNR}\cite{FK1}\cite{FK2}):
\begin{eqnarray}
\{x_1, x_2\}&=&r_{12}x_1x_2-x_1x_2r_{21}+x_1r_{21}x_2-x_2r_{12}x_1 \nonumber \\
\{x_1, y_2\}&=&-r_{21}x_1y_2-x_1y_2r_{21}+x_1r_{21}y_2-y_2r_{12}x_1 \\
\{y_1, y_2\}&=&r_{12}y_1y_2-y_1y_2r_{21}+y_1r_{21}y_2-y_2r_{12}y_1 \nonumber
\end{eqnarray}
Here $x$ and $y$ are matrix elements of $x\in G$ in a finite dimensional representation\footnote{These matrix elements for finite dimensional representations form a basis in the space of regular functions on $G$.}.
The matrix $r_{12}$ is the result of evaluation of the universal classical $r$-matrix from section
\ref{r-matrix} in the tensor product of two finite dimensional representations of $G$.

The phase space of the relativistic Calogero-Moser system is the symplectic leaf $\cM(\cC)$ of the moduli space
$H(G)/G$ corresponding to the conjugacy class $\cC$ of the monodromy $xyx^{-1}y^{-1}$ around the puncture.
In terms of Poisson geometry, this symplectic leaf can be described as follows.

The map $H(G)\to G$, $(x,y)\mapsto x$ is the $G$-valued moment  map for the left action of
the group on $H(G)$ (regarded as non-linear version of the cotangent bundle on $G$ trivialized by left
translations). The map $H(G)\to G$, $(x,y)\to yxy^{-1}$ is the group valued moment map for the corresponding right action of $G$. The map $\mu: H(G)\to G$, $\mu: (x,y)\mapsto xyx^{-1}y^{-1}$ is the group valued map corresponding
to the conjugation action. For details on group valued moment maps see \cite{A}.
Thus,
\[
\cM(\cC)=\mu^{-1}(\cC)/Ad_G\subset H(G)/Ad_G
\]
is the symplectic leaf corresponding to the conjugacy class $\cC$ of the monodromy around the puncture.

Hamiltonians of the {\it relativistic spin Calogero-Moser system} corresponding to the
conjugacy class $\cC$ are conjugation invariant functions on $G$, i.e. functions on $G/Ad_G$. The Hamiltonian
corresponding to $f\in C^G(G)$ is $H_f(x,y)=f(x)$.

The degenerate integrability of the relativistic spin Calogero-Moser system is
described by restricting the following sequences of Poisson maps:
\begin{equation}\label{rsr-rscm-degint}
(G\times G)/Ad_G\rightarrow (G\tilde{\times}_{G/Ad_G}G)/Ad_G\rightarrow G/Ad_G
\end{equation}
to the symplectic leaf $\cM(\cC)$. Here the fibered product is twisted as in the relativistic Casimir
system by $g\mapsto g^{-1}$ and $G(x,y)\mapsto G(x, yx^{-1}y^{-1})\mapsto Gx$.

This gives:
\[
\cM(\cC)\stackrel{\pi_1}{\rightarrow} \cP_1(\cC)\stackrel{p_1}{\rightarrow} \cB_1(\cC)\subset G/Ad_G
\]
where $\cP_1(\cC)=\{G(g_1,g_2)| G(g_1)=G(g_2^{-1}), \ \ g_1g_2\in \cC\}$ and $\cB_1(\cC)=\{ \cC'\in G/G| \cM_{\cC',{\cC'}^{-1},\cC}\neq \emptyset\}$. The space $\cM(\cC_1, \cC_2, \cC_3)$ is the moduli space of
flat $G$-connections on an oriented sphere with three punctures, with holonomies around punctures (in the
direction of the orientation of a sphere) being constrained to conjugacy classes $\cC_1, \cC_2, \cC_3$.
This are the same moduli spaces that appear in relativistic Casimir systems. Symplectic leaves of
$\cP_1(\cC, \cC')$ are $\cP_1(\cC, \cC')=\{G(g_1,g_2)| G(g_1)=G(g_2^{-1})=\cC'), \ \ G(g_1g_2)=\cC\}\simeq\cM_{\cC',{\cC'}^{-1},\cC}$.

Hamiltonians of the {\it relativistic spin Ruijsenaars system} are $H_f(x,y)=f(y)$ where
$f\in C^G(G)$ is a function on $G$, invariant with respect to conjugations.

The degenerate integrability of relativistic spin Ruijsenaars system is given by
restricting maps
\[
(G\times G)/Ad_G{\rightarrow} (G\times G)/Ad_G{\rightarrow} G/Ad_G
\]
to a symplectic leaf of $(G\times G)/Ad_G$. Here $G(x,y)\mapsto G(xyx^{-1}y^{-1}, y)\mapsto Gy$.
This gives:
\[
\cM(\cC)\stackrel{\pi_2}{\rightarrow} \cP_2(\cC)\stackrel{p_2}{\rightarrow} \cB_2(\cC)\subset G/Ad_G
\]
where  $\cP_2(\cC)=\{G(g_1,g_2)|g_1\in \cC, Gg_2=G(g_1g_2)\}$ and $\cB_2(\cC)=\{ \cC'\in G/G| \cM_{\cC',{\cC'}^{-1},\cC}\neq \emptyset\}$. Symplectic leaves of $P_2(\cC)$ are $P_2(\cC, \cC')=
\{G(g_1,g_2)|g_1\in \cC, Gg_2=G(g_1g_2)=\cC'\}\simeq \cM_{\cC',{\cC'}^{-1},\cC}$. This sequence of Poisson
projections describes the integrability of spin Ruijsenaars systems.

\subsection{Duality}
Spin Calogero-Moser and spin Ruijsenaars systems are related by the following transformation from the mapping class group of a torus:
\begin{proposition}
The mapping $\phi: G \times G\to G\times G$, $(x,y)\mapsto (y^{-1}, yxy^{-1})$
induces a Poisson map on $H(G)/Ad_G$. It induces the symplectomorphism $\cM(\cC)\mapsto \cM(\cC)$
that maps the relativistic spin Calogero-Moser system to relativistic spin Ruijsenaars system
and which is an equivalence of degenerate integrable systems.
\end{proposition}

\begin{proof} We shall complete the mapping $\phi$ to mappings $\phi_1$ and $\phi_2$ such that the
diagram (\ref{deg-int-equivalence}). Choose $\phi_1: \cP(\cC)\to \cP(cC)$  as
\[
G(g_1, g_2)\mapsto G(g_1g_2, g_2)
\]
and and $\phi_2=id$. It is easy to check that $\phi_1$ is Poisson. The map $\phi_2$ is
obviously Poisson. The commutativity of (\ref{deg-int-equivalence}) is obvious:
\[
\begin{tikzpicture}
\node (M) at (0, 5){$G(x,y)$};
  \node (P) at (0, 2.5) {$G(x, yx^{-1}y^{-1})$};
  \node (B) at  (0,0) {$G(x)$};
  \node (N) at (4,5) {$G(y^{-1},yxy^{-1}) $};
  \node (C) at (4,2.5) {$G(Pyx^{-1}y^{-1}x, yxy^{-1})$};
  \node (D) at (4,0) {$G(yxy^{-1})$};
  \path[->,font=\scriptsize,>=angle 90]
  (M) edge node[left] {$\pi$} (P)
  (M) edge node[above] {$\phi$} (N)
  (N) edge node[right]{$\pi'$} (C)
  (P) edge[above] node {$\phi_1$} (C)
  (P) edge[left] node {$p$} (B)
  (C) edge[right] node {$p'$} (D)
  (B) edge[above] node {$\phi_2$} (D);
\end{tikzpicture}
\]
Thus, maps $\phi$, $\phi_1$ and $\phi_2$ give the equivalence of degenerate integrable systems
between spin Calogero-Moser and spin Ruijsenaars systems.
\end{proof}

Let us prove that the two systems are dual in a sense of intersection property of Liouville tori.
This can be regarded as the duality between action-angle variables.

Let $\pi_1$ be the projection
\[
(G\times G)/Ad_G\rightarrow (G\widetilde{\times}_{G/G}G)/Ad_G, \ \ G(x,y)\mapsto G(x, yx^{-1}y^{-1})
\]
and $\pi_2$ be the projection
\[
(G\times G)/Ad_G\rightarrow (G\times G)/Ad_G, \ \  G(x,y)\mapsto G(xyx^{-1}y^{-1}, y),
\]
Denote fibers of these projections through the point $G(x,y)\in (G\times G)/Ad_G$ by $F_1(G(x,y))$ and $F_2(G(x,y))$ respectively.

\begin{proposition}
For generic $(x,y)$ we have:
\begin{enumerate}
\item $F_1(G(x,y))=\{G(x,yz)| z\in Z_x\}$ where $Z_x$ is the centralizer of $x$ in $G$.
\item $F_2(G(x,y))=\{G(xz,y)| z\in Z_y\}$
\item $F_1(G(x,y))\cap F_2(G(x,y))=G(x,y)$
\end{enumerate}
\end{proposition}

\begin{proof} First look at the fiber $F_1$:
\[
F_1(G(x,y))=\{G(x',y')| G(x,yx^{-1}y^{-1})= G(x',y'x'^{-1}y'^{-1})\}
\]
If $x'=gxg^{-1}$ the condition on $y'$ holds if and only if $y'=gyzg^{-1}$ where $zx=xz$. This proves the first statement.
The proof of the second statement is completely similar. Finally, it is clear that $G(xz,y)=G(x,y\tilde{z})$
where $z\in Z_y$ and $\tilde{z}\in Z_x$ if only if $z=\tilde{z}=1$.
\end{proof}

\subsection{Hamiltonians for rank 1 conjugacy classes in $SL_n$} Assume that $z=xyx^{-1}y^{-1}\in SL_n$ belongs to the rank 1 conjugacy class. This corresponds to spinless models. For generic rank 1 conjugacy class this means
\[
z=u \mbox{diag}(q^{n-1}, q,\dots, q) u^{-1}
\]
for some $u\in SL_n$ and $q\in \CC^*$. Equivalently, we can write
\[
z_{ij}=\phi_i\psi_j + q^{-1}\delta_{ij}
\]
where $(\phi, \psi)=\sum_{i=1}^n\psi_i\phi_i=q^{n-1}-q^{-1}$.

Hamiltonians of relativistic (nonspin) Calogero-Moser and relativistic Ruijsenaars systems are
\[
H_k^{rCM}=\chi_{\omega_k}(x), \ \ H_k^{sR}=\chi_{\omega_k}(y)
\]
Let us compute them in appropriate coordinates.

First, assume $x$ is semisimple and bring it to the diagonal form with eigenvalues $x_1, \dots, x_n$.
From the definition of $z$ we have
\begin{equation}\label{y-eq}
y_{ij}x_j=\sum_{k=1}^n x_ky_{kj}=\phi_i\sum_{k=1}^n \psi_kx_ky_{kj}-q^{-1}x_iy_{ij}
\end{equation}
From here we have:
\[
y_{ij}=\frac{\phi_i\sum_{k=1}^n\psi_kx_ky_{kj}}{x_j-q^{-1}x_i}
\]
Multiplying by $\phi_i\psi_i$ and taking sum over $i$ gives the following equation for
$\psi_i\phi_i$:
\[
\sum_{i=1}^n\frac{\psi_i\phi_ix_i}{x_j-q^{-1}x_i}=1
\]
Solving this equation we have
\[
\psi_i\phi_i=(1-q^{-1})x_i^{-1}\prod_{j\neq i}^n \frac{1-qx_jx_i^{-1}}{1-x_jx_i^{-1}}
\]
When $i=j$, (\ref{y-eq}) implies
\[
y_{ii}=\frac{\phi_i}{x_i(1-q^{-1})} \sum_k\psi_k x_k y_{ki}
\]
Solving this for $ \sum_k\psi_k x_k y_{ki}$ we have
\[
y_{ij}=\frac{\phi_i\phi_j^{-1}(1-q^{-1})y_{jj}}{1-q^{-1}x_ix_j^{-1}}
\]
Now we can compute Hamiltonians of the relativistic Ruijsenaars model in terms of $y_{ii}$ and $x_i$. For the first
two we have
\[
tr(y)=\sum_{j=1}^n y_{jj}
\]
\[
tr(y^2)=\sum_{ij}^n \frac{(1-q^{-1})^2y_{ii}y_{jj}}{(1-q^{-1}x_ix_j^{-1})(1-q^{-1}x_jx_i^{-1})}
\]
Poisson algebra $C(\cM(\cC)$ is isomorphic to the algebra of  symmetric Laurent polynomials
in $y_{ii}$ and $x_i$ (with respect to the diagonal action of the symmetric group)
with following Poisson brackets between $x$ and $y$:
\[
\{x_i, x_j\}=0, \ \ \{x_i, u_j\}=\delta_{ij}x_iu_j, \ \ \{u_i, u_j\}=0
\]
where
\[
y_{ii}=u_i\prod_{j\neq i}\frac{1-q^{-1}x_jx_i^{-1}}{1-x_jx_i^{-1}}
\]

The Hamiltonians $\chi_{\omega_i}(y)$ are classical analogs of  Macdonald operators. The
Hamiltonian of the relativistic Ruijsenaars model is
\[
H_2=\chi_{\omega_2}(y)=-q^{-1}\sum_{i<j}u_iu_j\prod_{a\in\{ij\},b\in \{ij\}^\vee}\frac{1-q^{-1}x_ax_b^{-1}}{1-x_ax_b^{-1}}
\]
The mapping $(x,y)\mapsto (y,x^{-1})$ intertwines the relativistic Calogero-Moser system and the relativistic
Ruijsenaars system. So, the Hamiltonian of relativistic Calogero-Moser model is given by essentially the same formula.

\section{Characteristic systems on simple Poisson Lie groups with standard Poisson Lie structure}

\subsection{Symplectic leaves and degenerate integrability of characteristic system }\label{r-matrix}
Standard Poisson Lie structure on a simple Lie group requires a choice of a Borel subgroup in $G$.
This fixes a Cartan subalgebra $\hh$, the root system and positive roots.
Assuming that the tangent bundle $TG$ is trivialized by left translations $TG\simeq \g\times G$, the Poisson bivector field
corresponding to the standard structure is
\[
\eta(x)=Ad_x(r)-r, \ \ r=\frac{1}{2}\sum_{i=1}^r h^i\otimes h_i+\sum_{\alpha>0} E_\alpha\otimes F_\alpha
\]
Here $\alpha$ are positive roots of $\g$, $E_\alpha; F_\alpha$ are corresponding elements of
the basis in $\g$; $r$ is the rank of $\g$; $h_i$ is a basis in the Cartan subalgebra $\hh$ and $h^i$ is the dual basis with respect to the Killing form. We assume that $\g\wedge \g\subset \g\otimes \g$.

Symplectic leaves of any Poisson Lie group are orbits of the dressing action
of the dual Poisson Lie group. For a simple Lie group $G$ with the standard Poisson Lie
structure symplectic leaves are known to be fibers of the fibration of
double Bruhat cells over tori inside of the Cartan subgroup $H$ of $G$. Recall that a double Bruhat cell in
$G$ is the intersection of a Bruhat cell for $B$ and a Bruhat cell for $B^-$:
\[
G^{u,v}=BuB\cap B^-vB^-
\]
where $BuB$ is defined as $B\overline{u}B\subset G$, where $u\in W$ and $\overline{u}\in N(H)\subset G$ is
its representative in the normalizer of $H$, and $B^-vB^-$ is defined similarly.

Generalized minors give a natural fibration
\[
\begin{tikzpicture}[scale=1.5]
\node (A) at (0,0) {$T^{u,v}$};
\node (B) at (0,1) {$G^{u,v}$};
\node (C) at (1,1) {$S^{u,v}$};
\draw[->] (C) to  (B);
\draw[->] (B) to (A);
\end{tikzpicture}
\]
For the explicit description of it see, for example,  \cite{R1} and references therein.

Hamiltonians of the characteristic integrable system are central functions on $G$.
There are only $r$ independent central functions which can be chosen as
characters of fundamental representations. Their restriction to a generic symplectic leave of
$G$ generates a degenerately integrable system \cite{R1}. Poisson projections describing degenerate integrability
can be described as follows:
\begin{equation}\label{degintproj}
S_{u,v}\to P^{u,v} \to Ad_GS_{u,v}.
\end{equation}
Here $P^{u,v}=(S^{u,v}\times S^{u,v})/Ad_{G^*}$ where $S^{u,v}\times S^{u,v}\subset G\times G$ and the
dual Poisson Lie group $G^*$ is embedded in $G\times G$ as usual $G^*=\{(b^+, b^-)\in B\times B^-\subset G\times G|[b^+]_0=[{b^-}^{-1}]_0\}$,
where $[b]_0$ is the Cartan component of $b\in B$.
The first map is the diagonal embedding, the second map is the projection to $(G\times G)/Ad_{G\times G}$
followed by the projection to any of the factors in the Cartesian product.

In other words, characteristic Hamiltonian systems are degenerately integrable
and their Liouville tori are intersections
of adjoint orbits of $G$ and of orbits of the dressing action of $G^*$.

\subsection{Hamiltonian flows as the factorization dynamics} Let $G$ be a factorizable Poisson-Lie group.
Note that the standard Poisson Lie group structure on a simple Lie group is an example
of a factorizable Poisson Lie group. Let $I(G)\subset C^\infty(G)$ be the subspace of $Ad_G$-invariant functions on $G$. For factorizable Poisson Lie groups $I(G)$ is a Poisson commutative Poisson algebra in $C^\infty(G)$.

Let $G^*$ be the dual Poisson Lie group to $G$. It has a natural embedding to $G\times G$ described above.
The multiplication in $G$, together with this embedding gives the mapping $G^*\to G$, $(b_+, b_-)\mapsto b_+b_-^{-1}$. When the inverse exists for the map $g\mapsto (g_+, g_-)$ (in a vicinity of the unit element in $G$ it is unique
when it exists), it is called the factorization map. Note that at the level of Lie algebras there
is always a linear isomorphism $\g\to \g^*$, such that $x=x_++x_0+x_-\mapsto (x_++\frac{x_0}{2}, -x_--\frac{x_0}{2})$.
It is called the factorization isomorphism.

The dynamics of characteristic systems can be described explicitly by the following theorem \cite{STS}:

\begin{theorem} \label{main}\hfill
Assume the factorization map is defined and unique on an open dense subset of $G$. Then
in a neighborhood of $t=0$ the flow lines of the Hamiltonian flow induced by $H\in I(G)$ passing through $x\in G$ at $t=0$ have the form
\[ x(t)=g_\pm(t)^{-1}x g_\pm(t), \]
where the mappings $g_\pm(t)$ are determined by
\[ g_+(t)g_-(t)^{-1}=\exp\left(tI\left(d_l H(x)\right)\right), \]
and $I:\g^*\longrightarrow\g$ is the inverse to the factorization isomorphism. Here $d_lH(x)\in \g^*$ is the left differential
of $H(x)$. For $X\in \g$, assuming the left trivialization of $TG$ we have $<d_lH(x), X>=\frac{d}{dt}H(e^{tX}x)|_{t=0}$ where $<.,.>: \g^*\times \g\to \CC$ is the natural pairing
(assuming we are over $\CC$).
\end{theorem}

\end{document}